\documentclass[12pt,a4paper]{article}

\usepackage{amsmath}
\usepackage{graphicx}
\usepackage{amssymb}
\usepackage{bold-extra}
\usepackage[linesnumbered,ruled,longend]{algorithm2e}
\usepackage{enumitem}
\usepackage{complexity}
\usepackage{amsthm}
\usepackage{fullpage}

\newtheorem{thm}{Theorem}[section]
\newtheorem{lma}[thm]{Lemma}

\newtheorem{prop}[thm]{Proposition}
\newtheorem*{adef}{Definition}

\DeclareMathOperator{\bin}{Bin}

\numberwithin{equation}{section}

\begin{document}

\newcommand{\ind}{\textbf{p-\#}\textsc{Independent Set}}
\newcommand{\mind}[0]{\textbf{p-\#}\textsc{Multicolour Independent Set}}
\newcommand{\mconn}[0]{\textbf{p-\#}\textsc{Multicolour Connected Subset}}
\newcommand{\conn}[0]{\textbf{p-\#}\textsc{Connected Subset}}
\newcommand{\leqfpt}{ $\leq^{\textrm{fpt}}$ }
\newcommand{\Q}{|\mathcal{Q}|}
\newcommand{\paramcount}[1]{\textup{\textbf{p-\#}}\textsc{#1}}
\newcommand{\paramdec}[1]{\textup{\textbf{p-}}\textsc{#1}}
\newcommand{\good}{good}

\title{Randomised enumeration of small witnesses using a decision oracle}
\author{Kitty Meeks\thanks{The author is supported by a Royal Society of Edinburgh Personal Research Fellowship, funded by the Scottish Government.}\\
\small{School of Computing Science, University of Glasgow} \\ 
\texttt{\small{kitty.meeks@glasgow.ac.uk}}}

\date{January 2018}

\maketitle

\begin{abstract}
Many combinatorial problems involve determining whether a universe of $n$ elements contains a witness consisting of $k$ elements which have some specified property.  In this paper we investigate the relationship between the decision and enumeration versions of such problems: efficient methods are known for transforming a decision algorithm into a search procedure that finds a single witness, but even finding a second witness is not so straightforward in general.  We show that, if the decision version of the problem can be solved in time $f(k) \cdot poly(n)$, there is a randomised algorithm which enumerates all witnesses in time $e^{k + o(k)} \cdot f(k) \cdot poly(n) \cdot N$, where $N$ is the total number of witnesses. If the decision version of the problem is solved by a randomised algorithm which may return false negatives, then the same method allows us to output a list of witnesses in which any given witness will be included with high probability.  The enumeration algorithm also gives rise to an efficient algorithm to count the total number of witnesses when this number is small.
\end{abstract}

\section{Introduction}

Many well-known combinatorial decision problems involve determining whether a universe $U$ of $n$ elements contains a witness $W$ consisting of \emph{exactly} $k$ elements which have some specified property; we refer to such problems as $k$-witness problems.  Although the decision problems themselves are of interest, it is often not sufficient for applications to output simply ``yes'' or ``no'': we need to \emph{find} one or more witnesses.  The issue of finding a single witness using an oracle for the decision problem has previously been investigated by Bj\"{o}rklund, Kaski, and Kowalik \cite{bjorklund14}, motivated by the fact that the fastest known parameterised algorithms for a number of widely studied problems (such as graph motif \cite{bjorklund13} and $k$-path \cite{bjorklund10}) are non-constructive in nature.  Moreover, for some problems (such as $k$-\textsc{Clique or Independent Set} \cite{arvind02} and $k$-\textsc{Even Subgraph} \cite{even}) the only known FPT decision algorithm relies on a Ramsey theoretic argument which says the answer must be ``yes'' provided that the input graph avoids certain easily recognisable structures.

Following the first approach used in \cite{bjorklund14}, we begin by assuming the existence of a deterministic inclusion oracle (a black-box decision procedure), as follows.
\\

\noindent
\textbf{INC-ORA($X$, $U$, $k$)}\\
\hangindent = \parindent
\textit{Input:} $X \subseteq U$ and $k \in \mathbb{N}$\\
\textit{Output:} 1 if some witness of size $k$ in $U$ is entirely contained in $X$; 0 otherwise.
\\

\noindent
Such an inclusion oracle can easily be obtained from an algorithm for the basic decision problem in the case of \emph{self-contained} $k$-witness problems, where we only have to examine the elements of a $k$-element subset (and the relationships between them) to determine whether they form a witness: we simply call the decision procedure on the substructure induced by $X$.  Examples of $k$-witness problems that are self-contained in this sense include those of determining whether a graph contains a $k$-vertex subgraph with some property, such as the well-studied problems $k$-\textsc{Path}, $k$-\textsc{Cycle} and $k$-\textsc{Clique}; algorithms to \emph{count} the number of witnesses in problems of this form have been designed for applications ranging from the analysis of biological networks \cite{milo02} to the design of network security tools \cite{gelbord01,sekar04,staniford-chen96}.

Given access to an oracle of this kind, a na\"{i}ve approach easily finds a single witness using $\Theta(n)$ calls to  \textbf{INC-ORA}: we successively delete elements of the universe, following each deletion with an oracle call, and if the oracle answers ``no'' we reinsert the last deleted element and continue.  Assuming we start with a yes-instance, this process will terminate when only $k$ elements remain, and these $k$ elements must form a witness.  In \cite{bjorklund14}, ideas from combinatorial group testing are used to make a substantial improvement on this strategy for the extraction of a single witness: rather than deleting a single element at a time, large subsets are discarded (if possible) at each stage.  This gives an algorithm that extracts a witness with only $2k\left(\log_2\left(\frac{n}{k}\right)+2\right)$ oracle queries.

However, neither of these approaches for finding a single witness can immediately be extended to find \emph{all} witnesses, a problem which is of interest even if an efficient decision algorithm does output a single witness.  Both approaches for finding a first witness rely on the fact that we can safely delete some subset of elements from our universe provided we know that what is left still contains at least one witness; if we need to look for a second witness, the knowledge that at least one witness will remain is no longer sufficient to guarantee we can delete a given subset.  Of course, for any $k$-witness problem we can check all possible subsets of size $k$, and hence enumerate all witnesses, in time $O(n^k)$; indeed, if \emph{every} set of $k$ vertices is in fact a witness then we will require this amount of time simply to list them all.  However, we can seek to do much better than this when the number of witnesses is small by making use of a decision oracle.

The enumeration problem becomes straightforward if we have an \emph{extension oracle},\footnote{Such an oracle is sometimes called an \emph{interval} oracle, as in the enumeration procedure described by Bj\"{o}rklund, Kaski, Kowalik and Lauri \cite{bjorklund15} which builds on earlier work by Lawler \cite{lawler72}.} defined as follows.
\\

\noindent
\textbf{EXT-ORA($X$,$Y$,$U$,$k$)}\\
\hangindent = \parindent
\textit{Input:} $X \subseteq U$, $Y \subseteq X$, and $k \in \mathbb{N}$\\
\textit{Output:} 1 if there exists a witness $W$ with $Y \subseteq W \subseteq X$; 0 otherwise.
\\

\noindent
The existence of an efficient procedure \textbf{EXT-ORA}($X$,$Y$,$U$,$k$) for a given self-contained $k$-witness problem allows us to use standard backtracking techniques to devise an efficient enumeration algorithm.  We explore a binary search tree of depth $O(n)$, branching at level $i$ of the tree on whether the $i^{th}$ element of $U$ belongs to the solution.  Each node in the search tree then corresponds to a specific pair $(X,Y)$ with $Y \subseteq X \subseteq U$; we can call \textbf{EXT-ORA($X$,$Y$,$U$,$k$)} to determine whether any descendant of a given node corresponds to a witness.  Pruning the search tree in this way ensures that no more than $O(n \cdot N)$ oracle calls are required, where $N$ is the total number of witnesses.

Note that, with only the inclusion oracle, we can determine whether there is a witness that does \emph{not} contain some element $x$ (we simply call \textbf{INC-ORA}($U \setminus \{x\}$, $U$, $k$)), but we cannot determine whether there is a witness which \emph{does} contain $x$.   Moreover, as we will show in Section \ref{sec:self-reduc}, there are natural (self-contained) $k$-witness problems for which the inclusion problem can be solved efficiently but there is no fpt-algorithm for the extension decision problem unless FPT=W[1].  This motivates the development of enumeration algorithms that do not rely on such an oracle.

The main result of this paper is just such an algorithm; specifically, we prove the following theorem.
\begin{thm}
There is a randomised algorithm to enumerate all witnesses of size $k$ in a $k$-witness problem exactly once, whose expected number of calls to a deterministic inclusion oracle is at most $e^{k + o(k)} \log^2 n \cdot N$, where $N$ is the total number of witnesses.  If an oracle call can be executed in time $g(k)\cdot n^{O(1)}$ for some computable function $g$, then the expected total running time of the algorithm is 
$$e^{k + o(k)} \cdot g(k) \cdot n^{O(1)} \cdot N.$$
Moreover, the total space required by the algorithm is at most $e^{k + o(k)} \cdot n^{O(1)}$.
\label{thm:enumeration}
\end{thm}
The key tool we use to obtain this algorithm is a colour coding method, using a family of $k$-perfect hash functions.  This technique was introduced by Alon, Yuster and Zwick in \cite{alon95} and has been widely used in the design of parameterised algorithms for decision and approximate counting (see for example \cite[Chapters 13 and 14]{flumgrohe} and \cite[Chapter 8]{downeyfellows13}), but to the best of the author's knowledge has not yet been applied to enumeration problems.

The main limitation of Theorem \ref{thm:enumeration} is that it requires access to a deterministic inclusion oracle \textbf{INC-ORA} which always returns the correct answer.  However, in a number of cases (including $k$-\textsc{Path} \cite{bjorklund10} and \textsc{Graph Motif} \cite{bjorklund13}) the fastest known decision algorithm for a self-contained $k$-witness problem (and hence for the corresponding inclusion problem) is in fact randomised and has a small probability of returning an incorrect answer.  We will also show that the same algorithm can be used in this case, at the expense of a small increase in the expected running time (if the oracle can return false positives) and the loss of the guarantee that we will output every witness exactly once: for each witness in the instance, there is a small probability that we will omit it from the list due to the oracle returning false negatives.  Specifically, we prove the following theorem.

\begin{thm}
Given a randomised inclusion oracle for the $k$-witness problem $\Pi$, whose probability of returning an incorrect answer is at most $c<\frac{1}{2}$, there is a randomised algorithm which takes as input an instance of $\Pi$ and a constant $\epsilon > 0$, and outputs a list of witnesses of size $k$ in the instance such that no witness appears more than once and, for any witness $W$, the probability that $W$ is included in the list is at least $1 - \epsilon$.  In expectation, the algorithm makes at most $e^{k + o(k)}\cdot \log(\epsilon^{-1}) \cdot \log^3 n \, (\log\log n) \cdot N$ oracle calls, where $N$ is the total number of witnesses in the instance, and if an oracle call can be executed in time $g(k) \cdot n^{O(1)}$ for some computable function $g$, then the expected running time of the algorithm is
$$e^{k + o(k)} \cdot \log(\epsilon^{-1} \cdot g(k) \cdot n^{O(1)} \cdot N.$$
Moreover, the total space required by the algorithm is $e^{k+o(k)} \cdot n^{O(1)}$.
\label{thm:enumeration-random}
\end{thm}

This result initiates the study of approximate algorithms for enumeration problems: in contrast with the well-established field of approximate counting, this relaxation of the requirements for enumeration does not seem to have been addressed in the literature to date.  

In the study of counting complexity it is standard practice, when faced with a $\#\P$-hard problem, to investigate whether there is an efficient method to solve the counting problem \emph{approximately}.  The answer to this question is considered to be ``yes'' if and only if the problem admits a \emph{fully polynomial randomised approximation scheme} (FPRAS), defined as follows.

\begin{adef}
An FPRAS for a counting problem $\Pi$ is a randomised approximation scheme that takes an instance $I$ of $\Pi$ (with $|I| = n$), and real numbers $\epsilon > 0$ and $0 < \delta < 1$, and in time $poly(n,1/\epsilon,\log(1/\delta))$ outputs a rational number $z$ such that
$$\mathbb{P}[(1-\epsilon)\Pi(I) \leq z \leq (1 + \epsilon)\Pi(I)] \geq 1 - \delta.$$
\end{adef}

\noindent
In the parameterised setting, the analogue of this is a \emph{fixed parameter tractable randomised approximation scheme} (FPTRAS), in which the running time is additionally allowed to depend arbitrarily on the parameter.

Perhaps the most obvious way to translate this notion in to the setting of enumeration would be to look for an algorithm which, with probability at least $(1 - \delta)$, would output at least $(1-\epsilon)$-proportion of all witnesses.  In the setting of counting, all witnesses are essentially interchangeable, so it makes sense to consider only the total number of objects counted in relation to the true answer.  However, this definition perhaps allows too much freedom in the setting of enumeration: we could design an algorithm which satisfies these requirements and yet will never output some collection of hard-to-find witnesses, so long as this collection is not too large compared with the total number of witnesses. 

Instead, we propose here a more demanding notion of approximate enumeration: given $\epsilon > 0$, we want a (randomised) algorithm such that, for any witness $W$, the probability we output $W$ is at least $1-\epsilon$.  This implies that we will, with high probability (depending on $\epsilon$) output a large proportion of all possible witnesses, but also ensures that we cannot choose to ignore certain potential witnesses altogether.  It may also be desirable to permit a witness to be repeated in the output with small probability: we can allow this flexibility by requiring only that, for each witness $W$, the probability that $W$ is included in the output exactly once is at least $1 - \epsilon$.  We give a formal definition of efficient approximate enumeration in Section \ref{sec:param}.

Theorem \ref{thm:enumeration} is proved in Section \ref{sec:algorithm}, and Theorem \ref{thm:enumeration-random} in Section \ref{sec:approx}.  We then discuss some implications of our enumeration algorithms for the complexity of related counting problems in Section \ref{sec:counting}.  We begin in Section \ref{sec:param} with some background on relevant complexity theoretic notions, before discussing the hardness of the extension version of some self-contained $k$-witness problems in Section \ref{sec:self-reduc}.

\section{Parameterised enumeration}
\label{sec:param}

There are two natural measures of the size of a self-contained $k$-witness problem, namely the number of elements $n$ in the universe and the number of elements $k$ in each witness, so the running time of algorithms is most naturally discussed in the setting of parameterised complexity.  There are two main complexity issues to consider in the present setting: first of all, as usual, the running time, and secondly the number of oracle calls required.

For general background on the theory of parameterised complexity, we refer the reader to \cite{downeyfellows13,flumgrohe}.  The theory of parameterised enumeration has been developed relatively recently \cite{fernau02,creignou13,creignou15}, and we refer the reader to \cite{creignou13} for the formal definitions of the different classes of parameterised enumeration algorithms.  To the best of the author's knowledge, this is the first occurrence of a randomised parameterised enumeration algorithm in the literature, and so we introduce randomised analogues of the four types of parameterised enumeration algorithms introduced in \cite{creignou13} (for a problem with total input size $n$ and parameter $k$, and with $f:\mathbb{N} \rightarrow \mathbb{N}$ assumed to be a computable function throughout): 
\begin{itemize}
\item an expected-total-fpt algorithm enumerates all solutions and terminates in expected time $f(k)\cdot n^{O(1)}$;
\item an expected-delay-fpt algorithm enumerates all solutions with expected delay at most $f(k) \cdot n^{O(1)}$ between the times at which one solution and the next are output (and the same bound applies to the time before outputting the first solution, and between outputting the final solution and terminating);
\item an expected-incremental-fpt algorithm enumerates all solutions with expected delay at most $f(k) \cdot (n + i)^{O(1)}$ between outputting the $i^{th}$ and $(i+1)^{th}$ solution;
\item an expected-output-fpt algorithm enumerates all solutions and terminates in expected time $f(k) \cdot (n+N)^{O(1)}$, where $N$ is the total number of solutions enumerated.
\end{itemize}

Under these definitions, Theorem \ref{thm:enumeration} says that, if there is an FPT decision algorithm for the inclusion version of a $k$-witness problem, then there is an expected-output-fpt algorithm for the corresponding enumeration problem.

In the setting of approximate enumeration, we define a \emph{fully output polynomial randomised enumeration scheme} (FOPRES) to be an algorithm which, given an instance $I$ of an enumeration problem (with total input size $n$) and a rational $\epsilon \in (0,1)$, outputs, in time bounded by a polynomial function of $n$, $N$ and $\epsilon^{-1}$ (where $N$ is the total number of solutions to $I$), a list of solutions to $I$ with the property that, for any solution $W$, the probability that $W$ appears exactly once in the list is at least $1 - \epsilon$.  In the parameterised setting, we analogously define a \emph{fully output fpt randomised enumeration scheme} (FOFPTRES) to be an algorithm which, given an instance $I$ of a parameterised enumeration problem (with total input size $n$ and parameter $k$) and a rational $\epsilon \in (0,1)$, outputs, in time bounded by $f(k)\cdot p(n,N,\epsilon^{-1})$, where $p$ is a polynomial, $f$ is any computable function, and $N$ is the total number of solutions to $I$, a list of solutions to $I$ with the property that, for any solution $W$, the probability that $W$ appears exactly once in the list is at least $1 - \epsilon$.  An \emph{expected}-FOPRES (respectively \emph{expected}-FOFPTRES) is a randomised algorithm which satisfies the definition of a FOPRES (resp. FOFPTRES) if we replace the condition on the running time by the same condition on the expected running time.  We can make analogous definitions for total-polynomial, total-fpt, delay-polynomial etc.

Under these definitions, Theorem \ref{thm:enumeration-random} says that, if there is a randomised FPT decision algorithm for the inclusion version of a $k$-witness problem with error probability less than a half, then the corresponding enumeration problem admits a FOFPTRES.

\section{Hardness of the extension problem}
\label{sec:self-reduc}

Many combinatorial problems have a very useful property, often referred to as \emph{self-reducibility}, which allows a search or enumeration problem to be reduced to (smaller instances of) the corresponding decision problem in a very natural way (see \cite{creignou13,khuller91,schnorr76}).  A problem is self-reducible in this sense if the existence of an efficient decision procedure (answering the question: ``Does the universe contain at least one witness of size $k$?") implies that there is an efficient algorithm to solve the extension decision problem (equivalent to \textbf{EXT-ORA}).  While many self-contained $k$-witness problems do have this property, we will demonstrate that there exist self-contained $k$-witness problems that do not (unless FPT=W[1]), and so an enumeration procedure that makes use only of \textbf{INC-ORA} and not \textbf{EXT-ORA} is desirable.

In order to demonstrate this, we show that there exist self-contained $k$-witness problems whose decision versions belong to FPT, but for which the corresponding extension decision problem is W[1]-hard.  We will consider the following problem, which is clearly a self-contained $k$-witness problem.
\\

\hangindent=\parindent
$k$-\textsc{Clique or Independent Set} \\
\textit{Input:} A graph $G = (V,E)$ and $k \in \mathbb{N}$.\\
\textit{Parameter:} $k$. \\
\textit{Question:} Is there a $k$-vertex subset of $V$ that induces either a clique or an independent set? \\

\noindent
This problem is known to belong to FPT \cite{arvind02}: any graph with at least $2^{2k}$ vertices must be a yes-instance by Ramsey's Theorem.  We now turn our attention to the extension version of the problem, defined as follows.
\\

\hangindent=\parindent
$k$-\textsc{Extension Clique or Independent Set} \\
\textit{Input:} A graph $G = (V,E)$, a subset $U \subseteq V$ and $k \in \mathbb{N}$.\\
\textit{Parameter:} $k$. \\
\textit{Question:} Is there a $k$-vertex subset $S$ of $V$, with $U \subseteq S$, that induces either a clique or an independent set? \\

\noindent
It is straightforward to adapt the hardness proof for $k$-\textsc{Multicolour Clique or Independent Set} \cite[Proposition 3.7]{treewidth} to show that $k$-\textsc{Extension Clique or Independent Set} is W[1]-hard.

\begin{prop}
$k$-\textsc{Extension Clique or Independent Set} is W[1]-hard.
\label{prop:ext-cliqueIS-hard}
\end{prop}
\begin{proof}
We prove this result by means of a reduction from the W[1]-complete problem $k$-\textsc{Clique}.  Let $(G,k)$ be the input to an instance of $k$-\textsc{Clique}.  We now define a new graph $G'$, obtained from $G$ by adding one new vertex $v$, and an edge from $v$ to every vertex $u \in V(G)$.  It is then straightforward to verify that $(G',\{v\},k+1)$ is a yes-instance for $k$-\textsc{Extension Clique or Independent Set} if and only if $G$ contains a clique of size $k$. 
\end{proof}

This demonstrates that $k$-\textsc{Extension Clique or Independent Set} is a problem for which there exists an efficient decision procedure but no efficient algorithm for the extension version of the decision problem (unless FPT=W[1]).  Both of these arguments (inclusion of the decision problem in FPT, and hardness of the extension version) can easily be adapted to demonstrate that the following problem exhibits the same behaviour.
\\

\hangindent=\parindent
$k$-\textsc{Induced Regular Subgraph} \\
\textit{Input:} A graph $G = (V,E)$ and $k \in \mathbb{N}$.\\
\textit{Parameter:} $k$. \\
\textit{Question:} Is there a $k$-vertex subset of $V$ that induces a subgraph in which every vertex has the same degree? \\

\noindent
Indeed, the same method can be applied to any problem in which putting a restriction on the degree of one of the vertices in the witness guarantees that the witness induces a clique (or some other induced subgraph for which it is W[1]-hard to decide inclusion in an arbitrary input graph).

\section{The randomised enumeration algorithm}
\label{sec:algorithm}

In this section we describe our randomised witness enumeration algorithm and analyse its performance when used with a deterministic oracle, thus proving Theorem \ref{thm:enumeration}.

As mentioned above, our algorithm relies on a colour coding technique.  A family $\mathcal{F}$ of hash functions from $[n]$ to $[k]$ is said to be \emph{$k$-perfect} if, for every subset $A \subset [n]$ of size $k$, there exists $f \in \mathcal{F}$ such that the restriction of $f$ to $A$ is injective.  We will use the following bound on the size of such a family of hash functions.

\begin{thm}\cite{naor95}
For all $n, k \in \mathbb{N}$ there is a $k$-perfect family $\mathcal{F}_{n,k}$ of hash functions from $[n]$ to $[k]$ of cardinality $e^{k + o(k)} \cdot \log n$.  Furthermore, given $n$ and $k$, a representation of the family $\mathcal{F}_{n,k}$ can be computed in time $e^{k + o(k)} \cdot  n \log n$.
\label{thm:hash-functions}
\end{thm}

Our strategy is to solve a collection of $e^{k + o(k)} \cdot \log n$ \emph{colourful enumeration problems}, one corresponding to each element of a family $\mathcal{F}$ of $k$-perfect hash functions.  In each of these problems, our goal is to enumerate all witnesses that are \emph{colourful} with respect to the relevant element $f$ of $\mathcal{F}$ (those in which each element is assigned a distinct colour by $f$).  Of course, we may discover the same witness more than once if it is colourful with respect to two distinct elements in $\mathcal{F}$, but it is straightforward to check for repeats of this kind and omit duplicate witnesses from the output.  It is essential in the algorithm that we use a deterministic construction of a $k$-perfect family of hash functions rather than the randomised construction also described in \cite{alon95}, as the latter method would allow the possibility of witnesses being omitted (with some small probability).

The advantage of solving a number of colourful enumeration problems is that we can split the problem into a number of sub-problems with the only requirement being that we preserve witnesses in which every element has a different colour (rather than all witnesses).  This makes it possible to construct a number of instances, each (roughly) half the size of the original instance, such that every colourful witness survives in at least one of the smaller instances.  More specifically, for each $k$-perfect hash function we explore a search tree: at each node, we split every colour-class randomly into (almost) equal-sized parts, and then branch to consider each of the $2^k$ combinations that includes one (nonempty) subset of each colour, provided that the union of these subsets still contains at least one witness (as determined by the decision oracle).  This simple pruning of the search tree will not prevent us exploring ``dead-ends'' (where we pursue a particular branch due to the presence of a non-colourful witness), but turns out to be sufficient to make it unlikely that we explore very many branches that do not lead to colourful witnesses.

We describe the algorithm in pseudocode (Algorithm \ref{alg:enumeration}), making use of two subroutines.  In addition to our oracle \textbf{INC-ORA}($X$,$U$,$k$), we also define a procedure \textbf{RANDPART}($X$) which we use, while exploring the search tree, to obtain a random partition of a subset of the universe.\\

\noindent
\textbf{RANDPART($X$)}\\
\hangindent = \parindent
\textit{Input:} $X \subseteq U$\\
\textit{Output:} A partition $(X_1,X_2)$ of $X$ with $\left| |X_1| - |X_2| \right| \leq 1$, chosen uniformly at random from all such partitions of $X$.
\\

\begin{algorithm}
\LinesNumbered
\If {$\mathrm{\mathbf{INC-ORA}}(U,U,k) = 1$}  {
Construct a family $\mathcal{F} = \{f_1,f_2,\ldots,f_{|\mathcal{F}|}\}$ of $k$-perfect hash functions from $U$ to $[k]$\; 

\For {$1 \leq r \leq |\mathcal{F}|$} { 
	Initialise an empty FIFO queue $Q$\;
	Insert $U$ into $Q$\;
	\While{$Q$ is not empty}{ 
		Remove the first element $A$ from $Q$\;
		\eIf {$|A| = k$} {
			\If {$A$ is not colourful with respect to $f_s$ for any $s \in \{1,\ldots,r-1\}$} {
				Output $A$\;
			}
		} 
		{  
			\For{$1 \leq i \leq k$}{
				Set $A_i$ to be the set of elements in $A$ coloured $i$ by $f_r$\;
				Set $(A_i^{(1)},A_i^{(2)}) = \mathrm{\mathbf{RANDPART}}(A_i)$\;
			}
			\For {each $\mathbf{j} = (j_1,\ldots,j_k) \in \{1,2\}^k$}{
				\If{$|A_i^{(j_{\ell})}| > 0$ for each $1 \leq \ell \leq k$} { 
					Set $A_{\mathbf{j}} = A_i^{(j_1)} \cup \cdots \cup A_i^{(j_k)}$\;
					\If {$\mathrm{\mathbf{INC-ORA}}(A_{\mathbf{j}},U,k) = 1$} { 
						Add $A_{\mathbf{j}}$ to $Q$\;
					}
				}
			}
		}
	} 
} 
}
 \caption{Randomised algorithm to enumerate all $k$-element witnesses in the universe $U$, using a decision oracle.}
 \label{alg:enumeration}
\end{algorithm}

\noindent
We prove the correctness of the algorithm and discuss the space used in Section \ref{sec:alg-correct}, and bound the expected running time in Section \ref{sec:alg-time}.

\subsection{Correctness of the algorithm}
\label{sec:alg-correct}

In order to prove that our algorithm does indeed output every witness exactly once, we begin by showing that we will identify a given $k$-element subset $X$ during the iteration corresponding to the hash-function $f \in \mathcal{F}$ if and only if $X$ is a colourful witness with respect to $f$.  

\begin{lma}
Let $X$ be a set of $k$ vertices in the universe $U$.  In the iteration of Algorithm~\ref{alg:enumeration} corresponding to $f \in \mathcal{F}$, we will execute 9 to 11 with $A = X$ if and only if:
\begin{enumerate}
\item $X$ is a witness, and
\item $X$ is colourful with respect to $f$.
\end{enumerate}
\label{lma:iff-colourful}
\end{lma}
\begin{proof}
We first argue that we only execute lines 9 to 11 with $A = X$ if $X$ is a witness and is colourful with respect to $f$.  We claim that, throughout the execution of the iteration corresponding to $f$, every subset $B$ in the queue $Q$ has the following properties:
\begin{enumerate}
\item there is some witness $W$ such that $W \subseteq B$, and
\item $B$ contains at least one vertex receiving each colour under $f$.
\end{enumerate}
Notice that we check the first condition before adding any subset $A$ to $Q$ (lines 1 and 20), and we check the second condition for any $A \neq U$ in line 18 ($U$ necessarily satisfies condition 2 by construction of $\mathcal{F}$), so these two conditions are always satisfied.  Thus, if we execute lines 9 to 11 with $A = X$, these conditions hold for $X$; note also that we only execute these lines with $A=X$ if $|X|=k$.  Hence, as there is a witness $W \subseteq X$ where $|W|=|X|=k$, we must have $X=W$ and hence $X$ is a witness.  Moreover, as $X$ must contain at least one vertex of each colour, and contains exactly $k$ elements, it must be colourful.

Conversely, suppose that $W = \{w_1,\ldots,w_k\}$ is a witness such that $f(w_i) = i$ for each $1 \leq i \leq k$; we need to show that we will at some stage execute lines 9 to 11 with $A = W$.  We argue that at the start of each execution of the while loop, if $W$ has not yet been output, there must be some subset $B$ in the queue such that $W \subseteq B$.  This invariant clearly holds before the first execution of the loop ($U$ will have been inserted into $Q$, as $U$ contains at least one witness $W$).  Now suppose that the invariant holds before starting some execution of the while loop.  Either we execute lines 9 to 11 with $A = W$ on this iteration (in which case we are done), or else we proceed to line 13.  Now, for $1 \leq i \leq k$, set $j_i$ to be either 1 or 2 in such a way that $w_i \in A_i^{(j_i)}$.  The subset $A_{\mathbf{j}}$, where $\mathbf{j} = (j_1,\ldots,j_k)$ will then pass both tests for insertion into $Q$, and $W \subseteq A_{\mathbf{j}}$ by construction, so the invariant holds when we exit the while loop.  Since the algorithm only terminates when $Q$ is empty, it follows that we must eventually execute lines 9 to 11 with $A = W$.
\end{proof}

The key property of $k$-perfect families of hash functions then implies that the algorithm will identify every witness; it remains only to ensure that we avoid outputting any witness more than once.  This is the purpose of lines 9 to 11 in the pseudocode.  We know from Lemma \ref{lma:iff-colourful} that we find a given witness $W$ while considering the hash-function $f$ if and only if $W$ is colourful with respect to $f$: thus, in order to determine whether we have found the witness in question before, it suffices to verify whether it is colourful with respect to any of the colourings previously considered.  Hence we see that every witness is output exactly once, as required.

Note that the most obvious strategy for avoiding repeats would be to maintain a list of all the witnesses we have output so far, and check for membership of this list; however, in general there might be as many as $\binom{n}{k}$ witnesses, so both storing this list and searching it would be costly.  The approach used here means that we only have to store the family $\mathcal{F}$ of $k$-perfect hash functions (requiring space $e^{k + o(k)} n \log n$).  Since each execution of the outer for loop clearly requires only polynomial space, the total space complexity of the algorithm is at most $e^{k + o(k)} n^{O(1)}$, as required.

\subsection{Expected running time}
\label{sec:alg-time}

We know from Theorem \ref{thm:hash-functions} that a family $\mathcal{F}$ of $k$-perfect hash functions from $U$ to $[k]$, with $|\mathcal{F}| = e^{k + o(k)}\log n$, can be computed in time $e^{k + o(k)}n \log n$; thus line 2 can be executed in time $e^{k + o(k)}n \log n$ and the total number of iterations of the outer for-loop (lines 2 to 34) is at most $e^{k + o(k)} \log n$.

Moreover, it is clear that each iteration of the while loop (lines 6 to 26) makes at most $2^k$ oracle calls.  If an oracle call can be executed in time $g(k)\cdot n^{O(1)}$ for some computable function $g$, then the total time required to perform each iteration of the while loop is at most $\max\{|\mathcal{F}|,kn + 2^k \cdot g(k)\cdot n^{O(1)}\} = e^{k + o(k)}\cdot g(k) \cdot n^{O(1)}.$

Thus it remains to bound the expected number of iterations of the while loop in any iteration of the outer for-loop; we do this in the next lemma.

\begin{lma}
The expected number of iterations of the while-loop in any given iteration of the outer for-loop is at most $N \left( 1 + \lceil \log n \rceil \right)$, where $N$ is the total number of witnesses in the instance.
\label{lma:while-bound}
\end{lma}
\begin{proof}
We fix an arbitrary $f \in \mathcal{F}$, and for the remainder of the proof restrict our attention to the iteration of the outer for-loop corresponding to $f$.  

We can regard this iteration of the outer for-loop as the exploration of a search tree, with each node of the search tree indexed by some subset of $U$.  The root is indexed by $U$ itself, and every node has up to $2^k$ children, each child corresponding to a different way of selecting one of the two randomly constructed subsets for each colour.  A node may have strictly fewer than $2^k$ children, as we use the oracle to prune the search tree (line 20), omitting the exploration of branches indexed by a subset of $U$ that does not contain any witness (colourful or otherwise).  Note that the search tree defined in this way has depth at most $\lceil \log n \rceil$: at each level, the size of each colour-class in the indexing subset is halved (up to integer rounding).

In this search tree model of the algorithm, each node of the search tree corresponds to an iteration of the while-loop, and vice versa.  Thus, in order to bound the expected number of iterations of the while-loop, it suffices to bound the expected number of nodes in the search tree.

Our oracle-based pruning method means that we can associate with every node $v$ of the search tree some representative witness $W_v$ (not necessarily colourful), such that $W_v$ is entirely contained in the subset of $U$ which indexes $v$.  (Note that the choice of representative witness for a given node need not be unique.)  We know that in total there are $N$ witnesses; our strategy is to bound the expected number of nodes, at each level of the search tree, for which any given witness can be the representative.

For a given witness $W$, we define a random variable $X_{W,d}$ to be the number of nodes at depth $d$ (where the root has depth 0, and children of the root have depth 1, etc.) for which $W$ could be the representative witness.  Since every node has some representative witness, it follows that the total number of nodes in the search tree is at most 
$$\sum_{W \text{ a witness}} \quad \sum_{d = 0}^{\lceil \log n \rceil} X_{W,d}.$$
Hence, by linearity of expectation, the expected number of nodes in the search tree is at most
$$\sum_{W \text{ a witness}} \quad \sum_{d = 0}^{\lceil \log n \rceil} \mathbb{E}\left[X_{W,d}\right] \quad \leq \quad N \sum_{d = 0}^{\lceil \log n \rceil} \max_{W \text{ a witness}} \mathbb{E}\left[X_{W,d}\right].$$

In the remainder of the proof, we argue that $\mathbb{E}[X_{W,d}] \leq 1$ for all $W$ and $d$, which will give the required result.

Observe first that, if $W$ is in fact a colourful witness with respect to $f$, then $X_{W,d} = 1$ for every $d$: given a node whose indexing set contains $W$, exactly one of its children will be indexed by a set that contains $W$.  So we will assume from now on that $W$ intersects precisely $\ell$ colour classes, where $\ell < k$.

If a given node is indexed by a set that contains $W$, we claim that the probability that $W$ is contained in the set indexing at least one of its children is at most $\frac{1}{2}^{k - \ell}$.  For this to happen, it must be that for each colour $i$, all elements of $W$ having colour $i$ are assigned to the same set in the random partition.  If $c_i$ elements in $W$ have colour $i$, the probability of this happening for colour $i$ is at most $\left(\frac{1}{2}\right)^{c_i - 1}$ (the first vertex of colour $i$ can be assigned to either set, and each subsequent vertex has probability at most $\frac{1}{2}$ of being assigned to this same set).  Since the random partitions for each colour class are independent, the probability that the witness $W$ survives is at most 
$$\prod_{W \cap f^{-1}(i) \neq \emptyset} \left( \frac{1}{2} \right)^{c_i - 1} = \left( \frac{1}{2} \right)^{k - |\{i: W \cap f^{-1}(i) \neq \emptyset\}|} = \left( \frac{1}{2} \right)^{k-\ell}.$$
Moreover, if $W$ is contained in the set indexing at least one of the child nodes, it will be contained in the indexing sets for exactly $2^{k - \ell}$ child nodes: we must select the correct subset for each colour-class that intersects $W$, and can choose arbitrarily for the remaining $k - \ell$ colour classes.  Hence, for each node indexed by a set that contains $W$, the \emph{expected} number of children which are also indexed by sets containing $W$ is at most $\left(\frac{1}{2} \right)^{k-\ell} \cdot 2^{k-\ell} = 1$.

We now prove by induction on $d$ that $\mathbb{E}\left[X_{W,d}\right] \leq 1$ (in the case that $W$ is not colourful).  The base case for $d = 0$ is trivial (as there can only be one node at depth $0$), so suppose that $d>0$ and that the result holds for smaller values.  Then, if $\mathbb{E}[Y|Z=s]$ is the conditional expectation of $Y$ given that $Z=s$,
\begin{align*}
\mathbb{E}[X_{W,d}] & = \sum_{t \geq 0} \mathbb{E}[X_{W,d}|X_{W,d-1} = t] \; \mathbb{P}[X_{W,d-1} = t] \\
					& \leq \sum_{t \geq 0} t \; \mathbb{P}[X_{W,d-1}=t] \\
					&  = \mathbb{E}[X_{W,d-1}]  \\
					& \leq 1,
\end{align*}
by the inductive hypothesis, as required.  Hence $\mathbb{E}[X_{W,d}] \leq 1$ for \emph{any} witness $W$, which completes the proof.
\end{proof}

By linearity of expectation, it then follows that the expected total number of executions of the while loop will be at most $|\mathcal{F}|\cdot N \left( 1 + \lceil \log n \rceil \right)$, and hence that the expected number of oracle calls made during the execution of the algorithm is at most $e^{k + o(k)} \log^2 n \cdot N$.  Moreover, if an oracle call can be executed in time $g(k)\cdot n^{O(1)}$ for some computable function $g$, then the expected total running time of the algorithm is 
$$e^{k + o(k)} \cdot g(k) \cdot n^{O(1)} \cdot N,$$
as required.

\section{Using a randomised oracle}
\label{sec:approx}

In this section we show that the method described in Section \ref{sec:algorithm} will in fact work almost as well if we only have access to a randomised decision oracle, thus proving Theorem \ref{thm:enumeration-random}.  The randomised decision procedures in \cite{bjorklund10,bjorklund13} only have one-sided errors, but for the sake of generality we consider the effect of both false positives and false negatives on our algorithm.

False positives and false negatives will affect the behaviour of the algorithm in different ways.  If the decision procedure gives false positives then, provided we add a check immediately before outputting a supposed witness that it really is a witness, the algorithm is still sure to output every witness exactly once; however, we will potentially waste time exploring unfruitful branches of the search tree due to false positives, so the expected running time of the algorithm will increase.  If, on the other hand, our algorithm outputs false negatives, then this will not increase the expected running time; however, in this case, we can no longer be sure that we will find every witness as false negatives might cause us to prune the search tree prematurely.  We will show, however, that we can still enumerate approximately in this case.

Before turning our attention to the specific effects of false positives and false negatives on the algorithm, we observe that, provided our randomised oracle returns the correct answer with probability greater than a half, we can obtain a decision procedure with a much smaller failure probability by making repeated oracle calls.  We make the standard assumption that the events corresponding to each oracle call returning an error are independent.

\begin{lma}
Let $c > \frac{1}{2}$ be a fixed constant, and let $\epsilon > 0$.  Suppose that we have access to a randomised oracle for the decision version of a self-contained $k$-witness problem which, on each call, returns the correct answer independently with probability at least $c$.  Then there is a decision procedure for the problem, making $O \left( k + \log \log n + \log \epsilon^{-1} \right)$ calls to this oracle, such that:
\begin{enumerate}
\item the probability of obtaining a false positive is at most $2^{-k}$, and
\item the probability of obtaining a false negative is at most $\frac{\epsilon}{\lceil \log n \rceil + 1}$
\end{enumerate}
\end{lma}
\begin{proof}
Our procedure is as follows: we make $t$ oracle calls (where $t$ is a value to be determined later) and output whatever is the majority answer from these calls.  We need to choose $t$ large enough to ensure that the probability that the majority answer is incorrect is at most $\delta :=\min \left\lbrace 2^{-k}, \frac{\epsilon}{\lceil \log n \rceil + 1}\right\rbrace$.

The probability that we obtain the correct answer from a given oracle call is at least $c$, so the number of correct answers we obtain out of $t$ trials is bounded below by the random variable $X$, where $X$ has distribution $\bin(t,c)$.  Thus $\mathbb{E}[X] = tc$.  We will return the correct answer so long as $X > \frac{t}{2}$.

Using a Chernoff bound, we can see that 
\begin{align*}
\mathbb{P}\left[ X \leq \frac{t}{2} \right] & = \mathbb{P} \left[ X \leq tc \cdot \frac{1}{2c} \right] \\
	&= \mathbb{P} \left[ X \leq tc \left( 1 - \frac{2c - 1}{2c}\right)\right] \\
	& \leq \exp \left(- \frac{1}{2}\left(\frac{2c - 1}{2c}\right)^2tc \right) \\
	& = \exp \left(- \frac{(2c - 1)^2t}{8c} \right).
\end{align*}
It is enough to ensure that this is at most $\delta$, which we achieve if
\begin{align*}
- \frac{(2c-1)^2t}{8c} & < \ln \delta \\
\Longleftrightarrow \qquad t & > \frac{- 8 c \ln \delta}{(2c - 1)^2},
\end{align*}
so we can take $t = O(\log \delta^{-1})$.  Thus the number of oracle calls required is 
\begin{align*}
O\left( \max \left\lbrace \log \left( 2^{-k}\right)^{-1}, \log \left(\frac{\epsilon}{\lceil \log n \rceil + 1}\right)^{-1} \right\rbrace \right) 
	&= O \left( k + \log \log n + \log \epsilon^{-1} \right),
\end{align*}
as required.
\end{proof}

We now show that, if the probability that our oracle gives a false positive is sufficiently small, then such errors do not increase the expected running time of Algorithm \ref{alg:enumeration} too much.  Just as when bounding the expected running time in Section \ref{sec:alg-time}, it suffices to bound the expected number of iterations of the while loop corresponding to a specific colouring $f$ in our family $\mathcal{F}$ of hash functions.

\begin{lma}
Suppose that the probability that the oracle returns a false positive is at most $\min \left \lbrace 2^{-k}, \frac{1}{\lceil \log n \rceil + 1} \right \rbrace$.  Then the expected number of iterations of the while-loop in any given iteration of the outer for-loop is at most $O(N \cdot \log^2n)$, where $N$ is the total number of witnesses in the instance.
\end{lma}
\begin{proof}
We fix an arbitrary $f \in \mathcal{F}$, and for the remainder of the proof we restrict our attention or the iteration of the outer for-loop corresponding to $f$.  As in the proof of Lemma \ref{lma:while-bound}, we can regard this iteration of the outer for-loop as the exploration of a search tree, and it suffices to bound the expected number of nodes in the search tree.  

We can associate with each node of the search tree some subset of the universe, and we prune the search tree in such a way that we only have a node corresponding to a subset $A$ of the universe if a call to the oracle with input $A$ has returned yes.  This means that for the node corresponding to the set $A$, either there is some representative witness $W \subseteq X$, or the oracle gave a false positive.  We call a node \emph{good} if it has some representative witness, and \emph{bad} if it is the result of a false positive.  We already bounded the expected number of good nodes in the proof of Lemma \ref{lma:while-bound}, so it remains to show that the expected number of bad nodes is not too large.

We will assume initially that there is at least one witness, and so the root of the search tree is a good node.  Now consider a bad node $v$ in the search tree; $v$ must have some ancestor $u$ in the search tree such that $u$ is good (note that the root will always be such an ancestor in this case).  Since the subset of the universe associated with any node is a subset of that associated with its parent, no bad vertex can have a good descendant.  Thus, any path from the root to the bad node $v$ must consist of a segment of good nodes followed by a segment of bad nodes; we can therefore associate with every bad node $v$ a unique good node $\good(v)$ such that $\good(v)$ is the last good node on the path from the root to $v$.  In order to bound the expected number of bad nodes in the tree, our strategy is to bound, for each good node $u$, the number of bad nodes $v$ such that $\good(v) = u$.

As in the proof of Lemma \ref{lma:while-bound}, we will write $X_{W,d}$ for the number of nodes at depth $d$ for which $W$ is the representative witness.  For $c > d$, we further define $Y_{W,d,c}$ to be the number of bad nodes $v$ such that $v$ is at depth $c$, $\good(v)$ is at depth $d$, and $W$ is the representative witness for $\good(v)$.

Since every node can have at most $2^k$ children, and the probability that the oracle gives a false positive is at most $2^{-k}$, the expected number of bad children of any node is at most one.  Thus we see that
\begin{align*}
\mathbb{E} \left[Y_{W,d,d+1}\right] & = \sum_{t \geq 0} \mathbb{E} \left[ Y_{W,d,d+1} | X_{W,d} = t \right] \mathbb{P} \left[ X_{W,d} = t \right] \\
	& \leq \sum_{t \geq 0} t \mathbb{P} \left[ W_{W,d} = t\right] \\
	& = \mathbb{E}\left[X_{W,d}\right].
\end{align*}
Observe also that if $u$ and $w$ are bad nodes such that $u$ is the child of $w$, then $\good(u) = \good(w)$ (and so $\good(u)$ and $\good(w)$ are at the same depth and have the same representative witness).  For $c > d+1$ we can then argue inductively:
\begin{align*}
\mathbb{E}\left[Y_{W,d,c}\right] &= \sum_{t \geq 0} \mathbb{E} \left[ Y_{W,d,c} | Y_{W,d,c-1} = t \right] \mathbb{P} \left[ Y_{W,d,c-1} = t \right] \\
	& = \sum_{t \geq 0} t \mathbb{P}\left[Y_{W,d,c-1} = t \right] \\
	& = \mathbb{E} \left[ Y_{W,d,c-1} \right] \\
	& = \mathbb{E} \left[ X_{W,d} \right].
\end{align*}
We can therefore bound the expected number of nodes in the search tree by
\begin{align*}
\mathbb{E} & \left[ \sum_{W \text{ a witness}} \sum_{d = 0}^{\lceil \log n \rceil} \left( X_{W,d} + \sum_{c=d+1}^{\lceil \log n \rceil} Y_{W,d,c} \right) \right] \\
& \qquad \qquad = \sum_{W \text{ a witness}} \sum_{d = 0}^{\lceil \log n \rceil} \left( \mathbb{E}\left[X_{W,d}\right] + \sum_{c=d+1}^{\lceil \log n \rceil} \mathbb{E} \left[ Y_{W,d,c} \right] \right) \\
& \qquad \qquad = \sum_{W \text{ a witness}} \sum_{d = 0}^{\lceil \log n \rceil} \left( \mathbb{E}\left[X_{W,d}\right] + \sum_{c=d+1}^{\lceil \log n \rceil} \mathbb{E} \left[ X_{W,d} \right] \right) \\
& \qquad \qquad = \sum_{W \text{ a witness}} \sum_{d = 0}^{\lceil \log n \rceil} \left(\lceil \log n \rceil - d + 1 \right) \mathbb{E}\left[X_{W,d}\right].
\end{align*}
As we know from the proof of Lemma \ref{lma:while-bound} that $\mathbb{E}\left[X_{W,d}\right] \leq 1$, we can therefore deduce that the expected number of nodes is at most
\begin{align*}
\sum_{W \text{ a witness}} \sum_{d = 0}^{\lceil \log n \rceil} \left(\lceil \log n \rceil - d + 1 \right) &= N \sum_{i=1}^{\lceil \log n \rceil + 1} i \\
& = \frac{N}{2} \left( \lceil \log n \rceil + 1 \right) \left( \lceil \log n \rceil + 2 \right)\\
& = O(N \log^2 n),
\end{align*}
as required.  This completes the proof in the case that the instance contains at least one witness.

If there is in fact no witness in the instance, we know that there are no good nodes in the tree.  Moreover, the expected number of bad nodes at depth $0$ is at most $1/\left(\lceil \log n \rceil + 1 \right)$ (the probability that the oracle returns a false positive).  Since we have already argued that the expected number of bad children of any node is at most $1$, it follows that the expected number of bad nodes at each level is at most $1/\left(\lceil \log n \rceil + 1 \right)$, and so the total expected number of bad nodes is at most $1/\left(\lceil \log n \rceil + 1 \right) \left(1 + \lceil \log n \rceil\right) = 1$.
\end{proof}

To complete the proof of Theorem \ref{thm:enumeration-random}, it remains to show that, so long as the probability that the oracle returns a false negative is sufficiently small, our algorithm will output any given witness with high probability.

\begin{lma}
Fix $\epsilon \in (0,1)$, and suppose that the probability that the oracle returns a false negative is at most $\frac{\epsilon}{\lceil \log n \rceil + 1}$.  Then, for any witness $W$, the probability that the algorithm does not output $W$ is at most $\epsilon$.
\end{lma}
\begin{proof}
By construction of $\mathcal{F}$, we know that there is some $f \in \mathcal{F}$ such that $W$ is colourful with respect to $f$.  We wil now restrict our attention to the iteration of the outer for-loop corresponding to $f$; it suffices to demonstrate that we will output $W$ during this iteration with probability at least $1 - \epsilon$.

If we obtain the correct answer from each oracle call, we are sure to output $W$.  The only way we will fail to output $W$ is if our oracle gives us an incorrect answer on at least one occasion when it is called with input $V \supseteq W$.  This can either happen in line 1 when we make the initial check that we have a yes-instance, or when we check whether a subset is still a yes-instance in line 20.  Note that we execute line 20 with $A_j = W$ at most $\lceil \log n \rceil$ times, so the total number of times we call \textbf{INC-ORA}($V$,$U$,$k$) with some $V \supseteq W$ during the iteration of the outer for-loop corresponding to $f$ is at most $\lceil \log n \rceil + 1$.  By the union bound, the probability that we obtain a false negative on at least one of these calls is at most 
$$\left( \lceil \log n \rceil + 1 \right) \cdot \frac{\epsilon}{\lceil \log n \rceil + 1} = \epsilon,$$
as required.
\end{proof}

\section{Application to counting}
\label{sec:counting}

There is a close relationship between the problems of counting and enumerating all witnesses in a $k$-witness problem, since any enumeration algorithm can very easily be adapted into an algorithm that counts the witnesses instead of listing them.  However, in the case that the number $N$ of witnesses is large, an enumeration algorithm necessarily takes time at least $\Omega(N)$, whereas we might hope for much better if our goal is simply to determine the total number of witnesses.

The family of self-contained $k$-witness problems studied here includes subgraph problems, whose parameterised complexity from the point of view of counting has been a rich topic for research in recent years \cite{flum04,connected,bddlayers,radu13,radu14,treewidth,even}.  Many such counting problems, including those whose decision problem belongs to FPT, are known to be \#W[1]-complete (see \cite{flumgrohe} for background on the theory of parameterised counting complexity).    Positive results in this setting typically exploit structural properties of the graphs involved (e.g. small treewidth) to design (approximate) counting algorithms for inputs with these properties, avoiding any dependence on $N$ \cite{alon08,arvind02,connected}.

In this section we demonstrate how our enumeration algorithms can be adapted to give efficient (randomised) algorithms to solve the counting version of a self-contained $k$-witness problem \emph{whenever the total number of witnesses is small}.  This complements the fact that a simple random sampling algorithm can be used for \emph{approximate} counting when the number of witnesses is very large \cite[Lemma 3.4]{treewidth}, although there remain many situations which are not covered by either result. 

We begin with the case in which we assume access to a deterministic oracle for the decision problem. 

\begin{thm}
Let $\Pi$ be a self-contained $k$-witness problem, and suppose that $0 < \delta \leq \frac{1}{2}$ and $M \in \mathbb{N}$.  Then there exists a randomised algorithm which makes at most $e^{k+o(k)} \log^2 n~M \log(\delta^{-1})$ calls to a deterministic decision oracle for $\Pi$, and
\begin{enumerate}
\item if the number of witnesses in the instance of $\Pi$ is at most $M$, outputs with probability at least $1 - \delta$ the exact number of witnesses in the instance;
\item if the number of witnesses in the instance of $\Pi$ is strictly greater than $M$, always outputs ``More than $M$.''
\end{enumerate}
Moreover, if there is an algorithm solving the decision version of $\Pi$ in time $g(k)\cdot n^{O(1)}$ for some computable function $g$, then the expected running time of the randomised algorithm is bounded by $e^{k + o(k)}\cdot g(k) \cdot n^{O(1)} \cdot M \cdot \log(\delta^{-1})$.
\label{thm:counting}
\end{thm}
\begin{proof}
Note that Algorithm \ref{alg:enumeration} can very easily be adapted to give a randomised counting algorithm which runs in the same time as the enumeration algorithm but, instead of listing all witnesses, simply outputs the total number of witnesses when it terminates.  We may compute explicitly the expected running time of our randomised enumeration algorithm (and hence its adaptation to a counting algorithm) for a given self-contained $k$-witness problem $\Pi$ in terms of $n$, $k$ and the total number of witnesses, $N$.  We will write $T(\Pi,n,k,N)$ for this expected running time.  

Now consider an algorithm $A$, in which we run our randomised counting algorithm for at most $2T(\Pi,n,k,M)$ steps; if the algorithm has terminated within this many steps, $A$ outputs the value returned, otherwise $A$ outputs ``FAIL''.  Since our randomised counting algorithm is always correct (but may take much longer than the expected time), we know that if $A$ outputs a numerical value then this is precisely the number of witnesses in our problem instance.  If the number of witnesses is in fact at most $M$, then the expected running time of the randomised counting algorithm is bounded by $T(\Pi,n,k,M)$, so by Markov's inequality the probability that it terminates within $2T(\Pi,n,k,M)$ steps is at least $1/2$.  Thus, if we run $A$ on an instance in which the number of witnesses is at most $M$, it will output the exact number of witnesses with probability at least $1/2$.

To obtain the desired probability of outputting the correct answer, we repeat $A$ a total of $\lceil\log(\delta^{-1})\rceil$ times.  If any of these executions of $A$ terminates with a numerical answer that is at most $M$, we output this answer (which must be the exact number of witnesses by the argument above); otherwise we output ``More than $M$.''

If the total number of witnesses is in fact less than or equal to $M$, we will output the exact number of witnesses unless $A$ outputs ``FAIL'' every time it is run.  Since in this case $A$ outputs ``FAIL'' independently with probability at most $1/2$ each time we run it, the probability that we output ``FAIL'' on every one of the $\lceil\log(\delta^{-1})\rceil$ repetitions is at most $(1/2)^{\lceil\log(\delta^{-1})\rceil} \leq 2^{\log \delta} = \delta$.  Finally, note that if the number of witnesses is strictly greater than $M$, we will always output ``More than $M$'' since every execution of $A$ must in this case return either ``FAIL'' or a numerical answer greater than $M$.

The total running time is at most $O\left(\log(\delta^{-1}) \cdot T(\Pi,n,k,M)\right)$ and hence, using the bound on the running time of our enumeration algorithm from Theorem \ref{thm:enumeration}, is bounded by $e^{k + o(k)}\cdot g(k) \cdot n^{O(1)} \cdot M \cdot \log(\delta^{-1})$, as required.
\end{proof}

Finally, we prove an analogous result in the case that we only have access to a randomised oracle.

\begin{thm}
Let $\Pi$ be a self-contained $k$-witness problem, suppose that $0 <\epsilon < 1$, $0 < \delta \leq \frac{1}{2}$ and $M \in \mathbb{N}$, and that we have access to a randomised oracle for the decision problem whose error probability is at most some constant $c < \frac{1}{2}$.  Then there exists a randomised algorithm which makes at most $e^{k + o(k)} \log^3 n~M \log(\delta^{-1})$ calls to this oracle and, with probability at least $1 - \delta$, if the total number of witnesses in the instance is exactly $N$, does the following:
\begin{enumerate}
\item if $N \leq M$, outputs a number $N'$ such that $(1-\epsilon)N \leq N' \leq N$;
\item if $N \geq M$, outputs either a number $N'$ such that $(1 - \epsilon) N \leq N' \leq M$ or ``More than $M$.''
\end{enumerate}
Moreover, if there is a randomised algorithm solving the decision version of $\Pi$ ~(with error probability at most $c < \frac{1}{2}$) in time $g(k)\cdot n^{O(1)}$ for some computable function $g$, then the expected running time of the randomised counting algorithm is bounded by $e^{k + o(k)}\cdot g(k) \cdot n^{O(1)} \cdot M \cdot \log(\delta^{-1})$.
\label{thm:counting-random}
\end{thm}
\begin{proof}
We claim that it suffices to demonstrate that there is a procedure which makes at most $e^{k + o(k)} \cdot \log^3 n \cdot M$ oracle calls and, with probability greater than $\frac{1}{2}$, outputs
\begin{enumerate}[label = (\alph*)]
\item a number $N'$ such that $(1 - \epsilon) N \leq N' \leq N$ if $N \leq M$, and
\item either a number $N'$ such that $(1-\epsilon)N \leq N' \leq N$ or ``FAIL'' if $N > M$.
\end{enumerate}
Given such a procedure, we run it $\log(\delta^{-1})$ times; if the largest numerical value returned on any run (if any) is at most $M$ then we return this maximum value, otherwise we return ``More than $M$.''  Conditions (a) and (b) ensuer that the procedure never returns a value strictly greater than $N$, so the largest numerical value returned (if any) is sure to be the best estimate.  Therefore we only return an answer that does not meet the conditions of the theorem if \emph{all} of the executions of the procedure fail to return an answer that meets conditions (a) and (b), which happens with probability at most $2^{- \log(\delta^{-1})} = \delta$.

To obtain the required procedure, we modify the enumeration algorithm used to prove Theorem \ref{thm:enumeration-random} so that it counts the total number of witnesses found rather than listing them; we will run this randomised enumeration procedure with error probability $\epsilon^2/4$.  We can compute explicitly the expected running time of this adapted algorithm for a given $k$-witness problem $\Pi$ in terms of $n$, $k$, $N$ and $\epsilon$; we write $T(\Pi, n, k, \epsilon, N)$ for this expectation.  We will allow the adapted algorithm to run for time $\displaystyle 4 T(\Pi,n,k,\epsilon,M)$, outputting ``FAIL'' if we have not terminated within this time.

There are two ways in which the procedure could fail to meet conditions (a) and (b).  First of all, the adapted enumeration algorithm might not terminate within the required time.  Secondly, it might terminate but with an answer $N'$ where $N' < (1-\epsilon) N$ (recall that enumeration algorithm never repeats a witness, and that we can verify each witness deterministically, ensuring that only ever output a subset of the witnesses actually present in the instance).  In the remainder of the proof, we will argue that the probability of each of these two outcomes is strictly less than $\frac{1}{4}$, so the probability of avoiding both is greater than $\frac{1}{2}$, as required.

First, we bound the probability that the algorithm does not terminate within the required time.  By Markov's inequality, the probability that a random variable takes a value greater than four times its expectation is less than $\frac{1}{4}$, so we see immediately that if the total number of witnesses is at most $M$ then the probability that the algorithm fails to terminate within the permitted time is less than $\frac{1}{4}$.  

Next, we need to bound the probability that the procedure outputs a value $N' < (1 - \epsilon) N$.  Let the random variable $Z$ denote the number of witnesses omitted by the procedure.  Then $\mathbb{E}[Z] \leq \epsilon^2 N/4$, so by Markov's inequality we have
$$\mathbb{P}[ Z > \epsilon N] \leq \frac{\epsilon^2 N /4}{\epsilon N} = \frac{\epsilon}{4} < \frac{1}{4},$$
as required.  This completes the argument that the procedure outputs the an answer that meets conditions (a) and (b) with probability greater than $\frac{1}{2}$, and hence the proof.
\end{proof}

\section{Conclusions and open problems}

Many well-known combinatorial problems satisfy the definition of the $k$-witness problems considered in this paper.  We have shown that, given access to a deterministic decision oracle for the inclusion version of a $k$-witness problem (answering the question ``does this subset of the universe contain at least one witness?''), there is a randomised algorithm which is guaranteed to enumerate all witnesses and whose expected number of oracle calls is at most $e^{k + o(k)} \log^2 n \cdot N$, where $N$ is the total number of witnesses.  Moreover, if the decision problem belongs to FPT (as is the case for many self-contained $k$-witness problems), our enumeration algorithm is an expected-output-fpt algorithm. 

We have also shown that, in the presence of only a randomised decision oracle, we can use the same strategy to produce a list of witnesses so that the probability of any given witness appearing in the list is at least $1 - \epsilon$, with only a factor $\log n$ increase in the expected running time.  This result initiates the study of algorithms for \emph{approximate} enumeration. 

Our results also has implications for counting the number of witnesses.  In particular, if the total number of witnesses is small (at most $f(k) \cdot n^{O(1)}$ for some computable function $f$) then our enumeration algorithms can easily be adapted to give fpt-algorithms that will, with high probability, calculate a good approximation to the number of witnesses in an instance of a self-contained $k$-witness problem (in the setting where we have a deterministic decision oracle, we in fact obtain the exact answer with high probability).  The resulting counting algorithms satisfy the conditions for a FPTRAS (Fixed Parameter Tractable Randomised Approximation Scheme, as defined in \cite{arvind02}), and in the setting with a deterministic oracle we do not even need the full flexibility that this definition allows: with probability $1 - \delta$ we will output the exact number of witnesses, rather than just an answer that is within a factor of $1 \pm \epsilon$ of this quantity.

While the enumeration problem can be solved in a more straightforward fashion for self-contained $k$-witness problems that have certain additional properties, we demonstrated that several self-contained $k$-witness problems do not have these properties, unless FPT=W[1].  A natural line of enquiry arising from this work would be the characterisation of those self-contained $k$-witness problems that do have the additional properties, namely those for which an fpt-algorithm for the decision version gives rise to an fpt-algorithm for the extension version of the decision problem. 

Our approach assumed the existence of an oracle to determine whether a given subset of the universe contains a witness of size \emph{exactly} $k$.  An interesting direction for future work would be to explore the extent to which the same techniques can be used if we only have access to a decision procedure that tells us whether some subset of the universe contains a witness of size \emph{at most} $k$.

Another key question that remains open after this work is whether the existence of an fpt-algorithm for the inclusion version of a $k$-witness problem is sufficient to guarantee the existence of an (expected-)delay-fpt or (expected-)incremental-fpt algorithm for the enumeration problem.  Finally, it would be interesting to investigate whether the randomised algorithm given here can be derandomised.

\end{document}